\documentclass[a4paper, 9pt, twosides]{amsart}

\usepackage{graphicx}
\usepackage{enumitem}
\usepackage{dsfont}
\usepackage[usenames, dvipsnames]{color}
\usepackage{todonotes}
\usepackage[colorlinks=true,
		raiselinks=true,
		linkcolor=MidnightBlue,
		citecolor=Mahogany,
		urlcolor=ForestGreen,
		pdfauthor={Atreyee Kundu and Debasish Chatterjee},
		pdftitle={Stabilizing switching signals},
		pdfkeywords={switched systems, switching signals},
		pdfsubject={Manuscript},
		plainpages=false]{hyperref}

\usepackage{newtxtext}
\usepackage[varg]{newtxmath}
\useosf
\usepackage[T1]{fontenc}

\DeclareFontFamily{T1}{pzc}{}
\DeclareFontShape{T1}{pzc}{m}{it}{<-> [1.2] pzcmi8t}{}
\DeclareMathAlphabet{\mathpzc}{T1}{pzc}{m}{it}

\usepackage{multicol}

\usepackage{mathtools}
\allowdisplaybreaks
\newtheorem{theorem}{Theorem}

\theoremstyle{definition}
\newtheorem{defn}{Definition}

\theoremstyle{remark}
\newtheorem{rem}{Remark}

\theoremstyle{assumption}
\newtheorem{assump}{Assumption}

\theoremstyle{fact}

\theoremstyle{claim}

\numberwithin{equation}{section}

\newcommand{\norm}[1]{\left\lVert{#1}\right\rVert}
\newcommand{\abs}[1]{\left\lvert{#1}\right\rvert}
\newcommand{\lra}{\longrightarrow}
\newcommand{\Let}{\coloneqq}
\newcommand{\teL}{\eqqcolon}

\newcommand{\pmat}[1]{\begin{pmatrix}#1\end{pmatrix}}

\renewcommand{\geq}{\geqslant}

\renewcommand{\leq}{\leqslant}
\renewcommand{\le}{\leqslant}

\newcommand{\R}{\mathbb{R}}
\newcommand{\N}{\mathbb{N}}
\renewcommand{\P}{\mathcal{P}}

\newcommand{\secref}[1]{\S\ref{#1}}

\newcommand{\KL}{\mathcal{KL}}
\newcommand{\Kinfty}{\mathcal{K}_{\infty}}
\providecommand{\ak[1]}{{\color{black}#1}}

\newcommand{\Nstsigma}{\mathrm{N}(s,t)} 
\newcommand{\Tj}{\mathrm{T}_{j}}
\newcommand{\Tk}{\mathrm{T}_{k}}
\newcommand{\TS}{\mathrm{{T}^S}}
\newcommand{\TU}{\mathrm{{T}^{U}}}

\newcommand{\PS}{\mathcal{P}_{S}}
\newcommand{\PU}{\mathcal{P}_{U}}

\newcommand{\Ntsigma}{\mathrm{N}(0,t)}

\title{Stabilizing switching signals: a transition from point-wise to asymptotic conditions}
\author{Atreyee Kundu}
\author{Debasish Chatterjee}
\thanks{Atreyee is with the Robert Bosch Centre for Cyber-Physical Systems, Indian Institute of Science, Bangalore - 560012, India, email: atreyee@cps.iisc.ernet.in. Debasish is with the Systems \& Control Engineering, Indian Institute of Technology Bombay, Powai, Mumbai - 400076, India, email: dchatter@iitb.ac.in.}


\date{\today}

\begin{document}

	\begin{abstract}
		Characterization of classes of switching signals that ensure stability of switched systems occupies a significant portion of the switched systems literature. This article collects a multitude of stabilizing switching signals under an umbrella framework. We achieve this in two steps: Firstly, given a family of systems, possibly containing unstable dynamics, we propose a new and general class of stabilizing switching signals. Secondly, we demonstrate that prior results based on both point-wise and asymptotic characterizations follow our result. This is the first attempt in the switched systems literature where these switching signals are unified under one banner. 
	\end{abstract}

    \maketitle


	\begin{multicols}{2}
	\mathtoolsset{mathic=true}
	\section{Introduction}
	\label{s:prob_state}
\subsection{The problem}
	\label{ss:prob_sett}
		A \emph{switched system} comprises of two components --- a family of systems and a switching signal. The \emph{switching signal} selects an \emph{active subsystem} at every instant of time, i.e., the system from the family that is currently being followed \cite[\S 1.1.2]{Liberzon}. Stability of switched systems is broadly classified into two categories --- \emph{stability under arbitrary switching} \cite[Chapter 2]{Liberzon} and \emph{stability under constrained switching} \cite[Chapter 3]{Liberzon}. In case of the former, conditions on the family of systems are determined such that the switched system generated under any admissible switching signal is stable; in case of the latter, given a family of systems, conditions on the switching signals are identified such that the resulting switched systems are stable. In this article we are interested in {identifying classes of stabilizing switching signals that ensure exponential convergence of switched systems in a sense to be made precise below.}

	We consider a family of continuous-time systems
        \begin{equation}
        \label{e:family}
            \dot{x}(t) = f_{i}(x(t)),\quad x(0) = x_{0},\quad i\in\P, \quad t\geq 0,
        \end{equation}
        where $x(t)\in\R^{d}$ is the vector of states at time $t$, and $\P = \{1,\ldots,N\}$ is an index set. We assume that for each $i\in\P$, $f_{i}:\R^{d}\lra\R^{d}$ is Lipschitz and $f_{i}(0) = 0$. Let $\sigma:[0,+\infty[\:\lra\P$ be a measurable function that specifies, at each time $t$, the index of the active system. The \emph{switched system} \cite[Chapter 1]{Liberzon} generated by the family of systems \eqref{e:family} and a fixed \emph{switching signal} $\sigma$ is given by
        \begin{equation}
        \label{e:swsys}
            \dot{x}(t) = f_{\sigma(t)}(x(t)),\quad x(0) = x_{0},\quad t\geq 0.
        \end{equation}
		Under the aforementioned assumptions on \(\sigma\), there exists \cite[Chapter 1]{Filippov} a Carath\'eodory solution of \eqref{e:swsys}. Let $0\teL \tau_{0}<\tau_{1}<\tau_{2}<\cdots$ denote the points of discontinuity of $\sigma$, henceforth called the \emph{switching instants}. Let $\Nstsigma$ denote the number of switches on an interval $]s,t]\subset[0,\infty[$. A switching signal is \emph{admissible} if it is piecewise constant as a function from $[0,+\infty[$ into $\P$, and by convention, is assumed to be continuous from {the} right and having limits from the left everywhere; we denote the set of admissible switching signals by $\mathcal{S}$.

        Given a family of systems \eqref{e:family}, we are interested in characterizing classes of switching signals \(\mathcal{S'}\subset\mathcal{S}\) such that for every $\sigma\in\mathcal{S'}$ the corresponding switched system \eqref{e:swsys} is globally asymptotically stable (GAS). Recall that:
        \begin{defn} \label{d:GAS}
            The switched system \eqref{e:swsys} is \emph{globally asymptotically stable} (GAS) for a given switching signal $\sigma$ if \eqref{e:swsys} is
            \begin{itemize}[label=\(\circ\), leftmargin=*]
                \item Lyapunov stable, and
                \item globally asymptotically convergent, i.e., irrespective of the initial condition $x_{0}$, we have $x(t)\to 0$ as $t\to+\infty$.
            \end{itemize}
           In other words, \eqref{e:swsys} is GAS for a given switching signal \(\sigma\) if there exists a class-\(\mathcal {KL}\) function \(\beta_\sigma\) such that \(\norm{x(t)} \le \beta_\sigma(\norm{x_0}, t)\) for all \(x_0\in\R^d\) and \(t\geq 0\).\footnote{Recall that
				$\mathcal{K} := \{\phi\::\:[0,+\infty[\to[0,+\infty[\:\:\big|\:\:\phi\:\:\text{ is continuous, strictly increasing},\:\:\phi(0) = 0\}$,
				$\KL := \{\phi\::\:[0,+\infty[^{2}\to[0,+\infty[\:\:\big|\:\:\phi(\cdot,s)\in\mathcal{K}\:\:\text{for each}\:\:s\:\:\text{and}\:\:\phi(r,\cdot)\searrow 0\:\:\text{as}\:\:s\nearrow+\infty\:\:\text{for each}\:\:r\}$,
				$\Kinfty := \{\phi\in\mathcal{K}\:\:\big|\:\:\phi(r)\rightarrow+\infty\:\:\text{as}\:\:r\rightarrow+\infty\}$.}
        \end{defn}

        \subsection{The basic assumptions}
	\label{the_assumps}

       	Let \(\P_{S}\) and \(\P_{U}\) denote the sets of indices of asymptotically stable and unstable systems in the family \eqref{e:family}, respectively, \(\P = \P_{S}\sqcup\P_{U}\). Let \(E(\P)\subset\P\times\P\) be the set of all ordered pairs \((i,j)\) such that the switching signal can jump from system \(i\) to system \(j\); in this case we say that the \emph{transition} \((i, j)\) is admissible.
	\begin{rem}
	\label{r:admissible_and_inadmissible_switches}
	\rm
		No distinction was made between admissible and inadmissible transitions, as we have defined above, in most of the classical works on switched systems. However, this distinction is becoming important in recent times; in particular, it plays a role in expressing situations where switches between certain subsystems may be prohibited. Such situations arise, for example, if it is known that switches from system $a$ to system $b$ are possible but not vice-versa, etc. In this article we employ a distinction between admissible and inadmissible transitions as described above, thereby allowing more descriptive specifications of switching signals --- clearly, ``unrestricted'' switching is a special case of restricted switching. In other words, if we construct a directed graph $G(\P,E(\P))$ in which the set of vertices $\P$ is the set of indices of the subsystems, and the set of directed edges $E(\P)$ defines the set of admissible transitions, the case of all transitions being admissible corresponds to the directed graph $G(\P,E(\P))$ being complete. 
	\end{rem}
	
	\begin{assump}
        \label{assump:key}
        There exist $\underline{\alpha}$, $\overline{\alpha}\in\mathcal{K}_{\infty}$, continuously differentiable functions $V_{i}:\R^{d}\:\lra [0,+\infty[$, $i\in\P$, and constants $\lambda_{i}\in\R$, $i\in\P$, such that
        \begin{align}
        \label{e:keyprop1}
            \underline{\alpha}(\norm{\xi})\leq V_{i}(\xi) \leq \overline{\alpha}(\norm{\xi})\quad \text{for all}\:\xi\in\R^{d},
        \end{align}
        and, with $\lambda_{i} > 0$ if \(i\in\P_{S}\) and $\lambda_{i} < 0$ if \(i\in\P_{U}\), we have for all $\gamma_{i}(0)\in\R^{d}$ and $t\in [0,+\infty[$,
        \begin{align}
            \label{e:Lyapprop}
            V_{i}(\gamma_{i}(t)) \leq V_{i}(\gamma_{i}(0))\exp(-\lambda_{i}t),
        \end{align}
        where $\gamma_{i}(\cdot)$ solves the $i$-th system dynamics in family \eqref{e:family}.
    \end{assump}

    	The functions \(V_{i}\)'s satisfying \eqref{e:keyprop1} and \eqref{e:Lyapprop} are called \emph{Lyapunov-like functions}, and they are standard in the literature, see e.g., \cite[Theorem 2]{Liberzon_IOSS}. The scalar $\lambda_{i}$ gives a quantitative measure of (in)stability of system $i\in\P$.

    \begin{assump}
        \label{assump:mu}
    For all $(i,j)\in E(\P)$, the respective Lyapunov-like functions are related as follows: there exists {$\mu_{ij} \geq 1$} such that
        \begin{equation}
        \label{e:muijprop}
            V_{j}(\xi)\leq\mu_{ij}V_{i}(\xi) \quad\text{for all}\:\: \xi\in\R^{d}.
        \end{equation}
    \end{assump}

    	The assumption of linearly comparable Lyapunov-like functions, i.e., there exists $\mu\geq1$ such that
        \begin{equation}
        \label{e:muprop}
            V_{j}(\xi)\leq\mu V_{i}(\xi)\quad\text{for all}\: \xi\in\R^{d}\:\text{and all}\: i,j\in \P,
        \end{equation}
        is standard in the theory of stability under average dwell time switching \cite[Theorem 3.2]{Liberzon}; \eqref{e:muijprop} gives more precise estimates than \eqref{e:muprop}.
	\subsection{A roughly chronological account of preceding works}
    	\label{ss:the_res}
		 Given a family of systems \eqref{e:family}, identification of classes of stabilizing switching signals primarily utilized the concept of ``slow switching'' vis-a-vis (\emph{average}) \emph{dwell time switching} \cite{Liberzon_survey, Antsaklis_survey}. Intuition suggests that a switched system whose constituent subsystems are all stable would itself be stable provided that the switching is ``slow''. Indeed, the basic idea of stability under slow switching is that if all the subsystems are stable and the switching is sufficiently slow, then the ``energy injected due to switching'' gets sufficient time for dissipation due to the stability of the individual subsystems. This idea is captured to some extent by the concepts of dwell time and average dwell time \cite[Chapter 3]{Liberzon}, \cite{Morse1996, HespanhaMorse, XieWenLi2001, chatterjee07}.

	 \subsubsection{(Average) Dwell time}
   	 \label{ss:dw_time}
	 	In the case of dwell time switching, a minimum duration of time is maintained between any two consecutive switching instants \cite[\S 3.2.1]{Liberzon}. Let us denote the \emph{i-th holding time} of a switching signal \(\sigma\) by
    \begin{align}
    \label{e:hold_time}
    	S_{i} := \tau_{i+1} - \tau_{i},\:\:i = 0,1,\ldots,
    \end{align}
    where \(\tau_{i}\) and \(\tau_{i+1}\) are two consecutive switching instants. A switching signal \(\sigma\) is said to satisfy a \emph{dwell time} \(\tau_{d} > 0\) if the inequality \(S_{i} \geq \tau_{d}\) is satisfied for all \(i = 0,1,\ldots\).

		Stabilizing dwell time switching was first proposed for switched linear systems in \cite[Lemma 2]{Morse1996}, and was later extended to the case of nonlinear systems in \cite{XieWenLi2001}; these signals are mostly of historical importance in the switched/hybrid systems community today.

	
		A more general class of switching signals, namely, those with an average dwell time \cite[\S3.2.2]{Liberzon}, allows the number of switches on any time interval to grow at most as an affine function of the length of the interval.
	The underlying idea is that stability of the switched system is preserved under fast switching, provided that the switches do not accumulate too quickly.
		
		A switching signal \(\sigma\) is said to satisfy an \emph{average dwell time} \(\tau_{a}\) if there exist \(\mathrm{N}_0, \tau_{a} > 0\) such that \(\displaystyle{\Nstsigma\leq \mathrm{N}_0 + \frac{t-s}{\tau_{a}}}\) for all \(]s,t]\subset[0,+\infty[\). The constant $\mathrm{N}_0$ is called a \emph{chatter bound}.
	
	Clearly, \(\sigma\) admits no switch if $\mathrm{N}_0$ is set to $0$, and if $\mathrm{N}_0 = 1$, a switching signal satisfying an average dwell time $\tau_{a}$ satisfies a dwell time $\tau_{d} = \tau_{a}$.

	\begin{theorem}[{\cite[Theorem 4]{HespanhaMorse}}]
	\label{t:avg_dw_time}
		Consider the family of systems \eqref{e:family} with $\PU = \emptyset$. Suppose that Assumption \ref{assump:key} holds with \(\abs{\lambda_{j}} = \lambda_{s}\) for all \(j\in\P_{S}\), and that Assumption \ref{assump:mu} holds with \(\mu_{ij} = \mu\) for all \((i,j)\in E(\P)\). 
Then the switched system \eqref{e:swsys} is GAS under every switching signal \(\sigma\in\mathcal{S}\) with an average dwell time
		\begin{align}
		\label{e:ADT_condn}
			\tau_{a} > \frac{\ln\mu}{\lambda_{s}}.
		\end{align}
	\end{theorem}
	Theorem \ref{t:avg_dw_time} has been widely employed in a diverse array of contexts within the switched systems literature {(\cite{chatterjee07, Liberzon_IOSS, Liberzon_fdrs})} and beyond \cite{Shim_adt_consensus, Tesi_TAC}. 
	At a first glance it may appear that the chatter bound \(\mathrm{N}_0\) provides an inexhaustible reserve of \(\mathrm{N}_0\) switches over every interval of time; indeed, the bound \(\Nstsigma \le \mathrm N_0 + \frac{t-s}{\tau_a}\) has to hold over every interval \(]s, t]\subset[0, +\infty[\). However, a closer inspection reveals that there is a reserve of only \(\mathrm{N}_0\) switches over the entire time axis \([0, +\infty[\) beyond the ones permissible for dwell time switching with \(\tau_d = \tau_a\).
\footnote{If by some time \(t' > 0\) we have \(\displaystyle{\mathrm{N}(0, t') = \mathrm{N_0} + \frac{t'}{\tau_a}}\), then for any \(k\in\{0, 1, \ldots\}\) and all \(s\in[k\tau_a, (k+1)\tau_a[\) we have \(\displaystyle{\mathrm{N}(0, t'+s) \le \mathrm{N_0} + \frac{(t'+s)}{\tau_a} = \mathrm{N_0} + \frac{t'}{\tau_a} + k}\).} The available number of reserve switches decreases every time that there is more than one switch on an interval of length \(\tau_a\). After these \(\mathrm{N}_0\) reserve switches are exhausted by fast switching, the average dwell time condition admits only one switch every \(\tau_a\) units of time --- i.e., it reduces to dwell time switching --- thereafter. Of course, there is no upper bound to \textsl{when} the reserve \(\mathrm{N}_0\) switches have to be exhausted.

	\subsubsection{Mode-dependent average dwell time}
	\label{ss:mdavg_dw_time}
		The class of stabilizing switching signals was further enlarged by the introduction of the concept of average dwell time specific to the subsystems in \cite{Zhao_TAC'12}. This variant of average dwell time is known as mode-dependent average dwell time.

		Let $\mathrm{N}_{{j}}(s,t)$ be the number of times a subsystem $j$ is activated on an interval $]s,t]\subset[0,+\infty[$, and $\mathrm{T}_{j}(s,t)$ denote the total duration of activation of the subsystem $j$ on $]s,t]$. In other words,
		 \begin{align}
		 \label{e:dur_defn}
		 	\mathrm{T}_{j}(s,t) = \abs{]s,t]\cap\biggl(\bigcup_{\displaystyle\substack{i=0\\\sigma(\tau_{i})=j}}^{+\infty}]\tau_{i},\tau_{i+1}]\biggr)},
		\end{align}
		where \(\abs{I}\) denotes the Lebesgue measure of a measurable set \(I\subset[0, +\infty[\). A switching signal $\sigma$ is said to satisfy a set of \emph{mode-dependent average dwell times} $\{\tau_{a}^{j}\}_{j\in\P}$ if there exist $\mathrm{N}_{0}^{j}$, $\tau_{{a}}^{j} > 0$ such that $\displaystyle{\mathrm{N}_{{j}}(s,t)\leq \mathrm{N}_{0}^{j} + \frac{\Tj(s,t)}{\tau_{a}^{j}}}$ for all $]s,t]\subset[0,+\infty[$, $j\in\P$. The constants $\{\mathrm{N}_{0}^{j}\}_{j\in\P}$ are called mode-dependent chatter bounds.

	\begin{theorem}[{\cite[Lemma 3]{Zhao_TAC'12}}]
	\label{t:mdavg_dw_time}
		Consider the family of systems \eqref{e:family} with $\PU = \emptyset$. Suppose that Assumption \ref{assump:key} holds, and Assumption \ref{assump:mu} holds with $\mu_{ij} = \mu_{j}$ for all $(i,j)\in E(\P)$. Then the switched system \eqref{e:swsys} is GAS for every switching signal $\sigma\in\mathcal{S}$ satisfying mode-dependent average dwell times
		\begin{align}
		\label{e:mdadtcondn1}
			\tau_{a}^{j} > \frac{\ln\mu_{j}}{\lambda_{j}},\quad j\in\P.
		\end{align}
	\end{theorem}
	
	{Theorem \ref{t:avg_dw_time} follows as a special case of Theorem \ref{t:mdavg_dw_time}}, and we shall see in \secref{s:proofs} how this implication holds.
	\subsubsection{Average dwell time with unstable subsystems}
	\label{ss:avg_dw_time+u}	
		So far we presented classes of stabilizing switching signals that cater to the family \eqref{e:family} containing all asymptotically stable systems. In the presence of unstable systems in the family, the preceding results do not carry over in a straightforward fashion. Indeed, slow switching alone cannot guarantee stability of switched systems when not all subsystems are asymptotically stable --- additional conditions are essential to ensure that the switched system does not spend too much time activating the unstable components \cite{Antsaklis_survey}. In \cite{Liberzon_IOSS} input/output-to-state stability (IOSS) of continuous-time switched systems such that not all subsystems are IOSS, was studied.\footnote{Recall that if both the input and the output map are set to $0$ for all time, then the IOSS property reduces to the GAS property.} It was shown that the switched system is IOSS under a class of switching signals satisfying a certain average dwell time and constrained point-wise activation of unstable subsystems.
		
		Let $\TS(s,t)$ and $\TU(s,t)$ denote the total durations of activation of the stable and the unstable subsystems on an interval $]s,t]\subset[0,+\infty[$, respectively. Clearly,
		\begin{align*}
			\TS(s,t) & = \sum_{j\in\P_{S}}\Tj(s,t),\\
			\TU(s,t) & = \sum_{k\in\P_{U}}\Tk(s,t),\text{ and}\\
			t-s 	 & = \TS(s,t) + \TU(s,t).
		\end{align*}
	
	\begin{theorem}[{\cite[Theorem 2]{Liberzon_IOSS}}]
	\label{t:avg_dw_time+u}
		Consider the family of systems \eqref{e:family}. Suppose that Assumption \ref{assump:key} holds with \(\abs{\lambda_{j}} = \lambda_{s}\) for all \(j\in\P_{S}\) and  \ak{\(\abs{\lambda_{k}} = \lambda_{u}\) for all \(k\in\P_{U}\)}, and Assumption \ref{assump:mu} holds with \(\mu_{ij} = \mu\) for all \((i,j)\in E(\P)\). 
		\ak{Let there exist constants \(\mathrm{T}_0\geq 0\) and \(\rho\in[0,\frac{\lambda_s}{\lambda_s+\lambda_u}[\) such that the following holds:}
		\begin{align}
		\label{e:un_durn}
			\TU(s,t) \leq \mathrm{T}_{0}+\rho(t-s)\quad\text{for every}\: ]s,t]\subset[0,+\infty[.
		\end{align}
		Then the switched system \eqref{e:swsys} is GAS under every switching signal \(\sigma\in\mathcal{S}\) satisfying an average dwell time
		\begin{align}
		\label{e:ADT+u}
			\tau_{a} > \frac{\ln\mu}{\lambda_{s}(1-\rho)-\lambda_{u}\rho}.
		\end{align}
	\end{theorem}
	
		For every interval $]s,t]\subset[0,+\infty[$ of time, the condition \eqref{e:un_durn} constrains the point-wise activation of unstable subsystems, while the average dwell time condition restricts the number of switches. It is evident that ${\rho} < 1$. The stabilizing class of switching signals is identified in terms of a (strict) lower bound on the average dwell time expressed in \eqref{e:ADT+u}. {Theorem \ref{t:avg_dw_time} follows as a special case of Theorem \ref{t:avg_dw_time+u} when $\PU = \emptyset$.}

		\begin{rem}
			{Theorems \ref{t:avg_dw_time}-\ref{t:avg_dw_time+u} employ multiple Lyapunov-like functions, and cater to the case of exponential convergence of the mapping \(t\mapsto V_{\sigma(t)}(x(t))\). In fact, they employ identical proof-techniques, with the later results refining some estimates that were employed in the preceding ones.}
		\end{rem}

	\subsubsection{Asymptotic conditions}
        \label{ss:asymp_condn}
        While the preceding efforts at characterizing stabilizing switching signals are related to point-wise properties of such signals, a sharp transition away from the prevailing trend appeared in the recent article \cite{abc}. This work dealt with switched systems with unstable subsystems, and provided a characterization of a class of stabilizing switching signals \emph{entirely} in terms of certain asymptotic properties, namely, the asymptotic frequency of switching, the asymptotic fraction of activity of the constituent subsystems, and the asymptotic ``density'' of the admissible transitions among them.

		We now define the succinct notations necessary for the above mentioned ``asymptotic'' condition. 
		Fix $t > 0$.\footnote{The premise of \cite{abc} is a bit more general, and the properties of a switching signal there are measured with respect to a class $\Kinfty$ function $h:[0,+\infty[\to[0,+\infty[$. In this article we keep $h(t)=t$ for simplicity.}. Let
		\begin{align}
		\label{e:swfreq}
			\displaystyle{\nu(t) \Let \frac{\Ntsigma}{t}}
		\end{align}
		be the \emph{frequency of switching} at $t$. We denote by $\mathrm{N}_{ij}(0,t)$ the number of times a switch from subsystem $i$ to subsystem $j$ has occurred before (and including) time $t$. It follows that $\displaystyle{\Ntsigma = \sum_{(i,j)\in E(\P)}\mathrm{N}_{ij}(0,t)}$. Let
		\begin{align}
		\label{e:tranfreq}
			 \rho_{ij}(t) := \frac{\mathrm{N}_{ij}(0,t)}{\Ntsigma}
		\end{align}
		be the \emph{transition frequency from subsystem $i$} to subsystem $j$ on $]0,t]$, $(i,j)\in E(\P)$.
		We let
		\begin{align}
		\label{e:actifrac}
			\eta_{j}(t) := \frac{\mathrm{T}_{j}(0,t)}{t}
		\end{align}
		denote the \emph{fraction of activation of subsystem $j$} on the interval $]0,t]$.
			
		\begin{theorem}[{\cite[Theorem 5]{abc}}]
		\label{t:asym_condn}
			Consider the switched system \eqref{e:swsys}. Let Assumptions \ref{assump:key} and \ref{assump:mu} hold. Then the switched system \eqref{e:swsys} is GAS under every switching signal $\sigma\in\mathcal{S}$ satisfying
			\begin{equation}
			\label{e:asympcondn2}
			\begin{aligned}
				& \varlimsup_{t\to+\infty}\nu(t)\sum_{(k,\ell)\in E(\P)}(\ln\mu_{k\ell})\varlimsup_{t\to+\infty}\rho_{k\ell}(t)\\
				& \qquad < \sum_{j\in\PS}\abs{\lambda_{j}}\varliminf_{t\to+\infty}\eta_{j}(t) - \sum_{k\in\PU}\abs{\lambda_{k}}\varlimsup_{t\to+\infty}\eta_{k}(t).
			\end{aligned}
			\end{equation}
		\end{theorem}
		
		Observe some of the key differences between Theorems \ref{t:avg_dw_time} and \ref{t:mdavg_dw_time} and Theorem \ref{t:asym_condn}: the first two relied on point-wise conditions on the number of switches on every interval of time, but the last utilizes only certain asymptotic properties of the switching signal. 
		The term on the left-hand side of \eqref{e:asympcondn2} is a product of the upper asymptotic density of the switching frequency $\nu$ and the factor $\displaystyle{\sum_{(k,\ell)\in E(\P)}(\ln\mu_{k\ell})\varlimsup_{t\to+\infty}{\rho}_{k\ell}}$, which contains the asymptotic upper density of ${\rho}_{k\ell}$, the frequency of admissible transitions among the systems in the given family \eqref{e:family}. The two terms on the right-hand side of \eqref{e:asympcondn2} involve the switching destinations. The first (\emph{resp}.\ second) term comprises of the lower (\emph{resp}.\ upper) asymptotic density of the total fraction of activation of the asymptotically stable (\emph{resp}.\ unstable) systems in \eqref{e:family}, weighted by the corresponding quantitative measures of (in)stability.
		
		The condition \eqref{e:asympcondn2} allows $\Ntsigma$ to grow faster than an affine function of $t$; indeed, \(\sigma\)'s with $\Ntsigma$ satisfying $k_{0}t - k_0'\sqrt{t} \leq\Ntsigma \le k_{1}+k'_{1}t+k''_{1}\sqrt{t}$ for positive constants $k_{0}$, $k'_{0}$, $k_{1}$, $k'_{1}$, \(k''_1\), are admissible. However, Theorem \ref{t:asym_condn} does not guarantee uniform stability in the sense of \cite[\S2]{Liberzon}. This inherent deficiency is, of course, only natural since \eqref{e:asympcondn2} does neither consider nor constrain the transient behaviour of the switching signals.
		
		\begin{rem}
			Observe that in \cite[Theorem 5]{abc} there is an additional condition \(\displaystyle{\varliminf_{t\to+\inf}\nu(t) > 0}\). However, this condition turns out to be superfluous, see Appendix for a detailed discussion on this matter.
		\end{rem}

		\begin{rem}
			{
				A glance at the proof of Theorem \ref{t:asym_condn} given in \cite{abc} reveals that this result also caters to the case of exponential convergence of the function \(t\mapsto V_{\sigma(t)}(x(t))\), much like the preceding Theorems \ref{t:avg_dw_time}-\ref{t:avg_dw_time+u}. However, it is interesting to note that the assertion of Theorem \ref{t:avg_dw_time+u} does not follow from Theorem \ref{t:asym_condn} when Theorem \ref{t:asym_condn} is specialized to the case of switching signals satisfying the conditions of Theorem \ref{t:avg_dw_time+u}. {For instance, consider a switched linear system \(\dot{x}(t) = A_{\sigma(t)}x(t)\), \(x(0) = x_0\), \(t\geq 0\). Let \(\P = \{1,2,3\}\) with \(A_1 = \pmat{-0.3 & 1\\-0.9 & -1.2}\), \(A_2 = \pmat{0.2 & 0.1\\0.3 & 0}\) and \(A_3 = \pmat{0.1 & 0.2\\0.3 & 0.1}\). Clearly, we have \(\P_S = \{1\}\) and \(\P_U = \{2,3\}\). Let \(E(\P) = \{(1,2),(1,3),(2,1),(2,3),(3,1),(3,2)\}\). We compute \(\lambda_j\), \(j\in\P\) and \(\mu_{k\ell}\), \((k,\ell)\in E(\P)\) from the estimates provided in \cite{def} and obtain: \(\lambda_1 = 0.9389\), \(\lambda_2 = -0.7301\), \(\lambda_3 = -0.7206\), \(\mu_{12} = \mu_{13} = 2.0611\), \(\mu_{21} = \mu_{31} = 1.0651\), \(\mu_{23} = \mu_{32} = 1\). We choose \(\lambda_s = \lambda_1=0.9389\), \(\lambda_u = \max{\{\lambda_2,\lambda_3\}} = 0.7301\) and \(\displaystyle{\mu = \max_{(i,j)\in E(\P)}\mu_{k\ell} = 2.0611}\). Now, consider a switching signal \(\sigma\) that satisfies Theorem \ref{t:avg_dw_time+u} with \(\mathrm{N}_0 = 2\), \(\mathrm{T}_0 = 0.3\), \(\rho = 0.55\), \(\tau_a = 6.93\). Let \(\mathrm{T}_2(0,t) = 0.25t\), \(\mathrm{T}_3(0,t) = 0.3t\), and \(\mathrm{N}_{k\ell}(0,t) = \frac{1}{6}\mathrm{N}(0,t)\) for all \((k,\ell)\in E(\P)\). Clearly, for the above \(\sigma\), \(\displaystyle{\varlimsup_{t\to+\infty}\nu(t)\sum_{(k,\ell)\in E(\P)}(\ln\mu_{k\ell})\varlimsup_{t\to+\infty}\rho_{k\ell}(t) - \sum_{j\in\P_S}\abs{\lambda_j}\varliminf_{t\to+\infty}\eta_j(t)}\) \(\displaystyle{+ \sum_{k\in\P_U}\abs{\lambda_k}\varliminf_{t\to+\infty}\eta_k(t)} = 0.0843\). Hence, \eqref{e:asympcondn2} is not satisfied. In view of the above example, the quest for a unifying framework capturing Theorems \ref{t:avg_dw_time} - \ref{t:asym_condn} is, therefore, only natural, and we establish such a framework in \secref{s:mainres}.}
			}
		\end{rem}
		
\subsection{Our contributions}
    	\label{ss:our_obj}
   		We have so far collected, in a roughly chronological order of appearance, various classes of stabilizing switching signals for continuous-time switched systems. The corresponding stability conditions are derived with the aid of multiple Lyapunov-like functions \cite[Chapter 3]{Liberzon}, and provide only ``sufficient'' conditions. In fact, the proof techniques of all the above results are essentially similar modulo minor differences. The switching signals in Theorems \ref{t:avg_dw_time}--\ref{t:avg_dw_time+u} are characterized based on their point-wise properties, while the characterization in Theorem \ref{t:asym_condn} relies solely on their asymptotic behaviour.  Theorems \ref{t:avg_dw_time+u}--{\ref{t:asym_condn}} cater to families in which not all systems are asymptotically stable, while Theorems \ref{t:avg_dw_time}--\ref{t:mdavg_dw_time} apply to families in which all systems are asymptotically stable. On the one hand, given a family of systems, numerically constructing a switching signal that satisfies certain conditions on every interval of time is a difficult task. On the other hand, stabilizing switching signals characterized on the basis of asymptotic behaviour of the switching signals afford a relatively simpler algorithmic synthesis, but fail to guarantee ``uniformity'' properties unlike the ones that satisfy point-wise conditions.\footnote{By ``uniformity'', here we mean uniformity over a class of switching signals satisfying certain conditions. To wit, suppose that there are two switching signals $\sigma_{1}$ and $\sigma_{2}$ that satisfy the conditions in Theorem \ref{t:asym_condn}. We have that under both $\sigma_{1}$ and $\sigma_{2}$, the switched system \eqref{e:swsys} is GAS with the corresponding class $\mathcal{KL}$ functions being $\beta_{\sigma_1}$ and $\beta_{\sigma_2}$, respectively. However, Theorem \ref{t:asym_condn} does not guarantee that $\beta_{\sigma_1} = \beta_{\sigma_2}$.}

		In the next section we propose a general framework that unifies all the preceding classes of stabilizing switching signals under one banner. We achieve this in two steps: Given a family of systems, in the first step, we identify a general class of stabilizing switching signals in Theorem \ref{t:mainres1}. Multiple Lyapunov-like functions are employed in our analysis, and the proposed class is characterized solely in terms of certain asymptotic quantities. In the second step (Theorem \ref{t:mainres2}), we show that all the classes of stabilizing switching signals that we have described above are unified by the one that we described in Theorem \ref{t:mainres1}. At this point it is important to clarify what we mean by ``unify'': we show that if a switching signal $\sigma$ satisfies the conditions in Theorem \ref{t:avg_dw_time} (resp.\ Theorems \ref{t:mdavg_dw_time}, \ref{t:avg_dw_time+u}, \ref{t:asym_condn}), then the conditions in Theorem \ref{t:mainres1} follow, and hence, by the assertion of Theorem \ref{t:mainres1}, the switched system \eqref{e:swsys} is GAS. Thus, we unify a large class of stabilizing switching signals under one umbrella framework.

\section{A unifying framework}
	\label{s:mainres}
		The first result of this article, Theorem \ref{t:mainres1} below, characterizes a broad class of stabilizing switching signals:
		\begin{theorem}
		\label{t:mainres1}
		{Consider the family of systems \eqref{e:family}. Let Assumptions \ref{assump:key} and \ref{assump:mu} hold. Then the switched system \eqref{e:swsys} is GAS for every switching signal \(\sigma\in\mathcal{S}\) that satisfies
		\begin{align}
		\label{e:maincondn2}
				\varlimsup_{t\to+\infty}\Biggl(\nu(t)\sum_{(k,\ell)\in E(\P)}(\ln\mu_{k\ell})\rho_{k\ell}(t) &- \sum_{j\in\PS}\abs{\lambda_{j}}\eta_{j}(t) \nonumber\\&\hspace*{-1.5cm}+ \sum_{k\in\PU}\abs{\lambda_{k}}\eta_{k}(t)\Biggr) <  0,
		\end{align}
    where $\lambda_{j}$, $j\in\P_{S}$, $\lambda_{k}$, $k\in\P_{U}$ and $\mu_{k\ell}$, $(k,\ell)\in E(\P)$ obey \eqref{e:Lyapprop} and \eqref{e:muijprop}, respectively, and $\nu(t)$, $\rho_{k\ell}(t)$, $(k,\ell)\in E(\P)$ and $\eta_{j}(t)$, $j\in\P_{S}$, $\eta_{k}(t)$, $k\in\P_{U}$ are as defined in \eqref{e:swfreq}, \eqref{e:tranfreq} and \eqref{e:actifrac}, respectively.}
	\end{theorem}

	The condition \eqref{e:maincondn2} determines the asymptotic nature of the function
	\begin{multline*}
		t\mapsto \nu(t)\sum_{(k,\ell)\in E(\P)}(\ln\mu_{k\ell})\rho_{k\ell}(t) \\
			- \sum_{j\in\PS}\abs{\lambda_{j}}\eta_{j}(t) + \sum_{k\in\PU}\abs{\lambda_{k}}\eta_{k}(t).
	\end{multline*}
	The first term in the expression of the preceding function includes the switching frequency and the transition frequency between subsystems, while the last two terms involve the fractions of activation of the subsystems. As in the case of Theorem \ref{t:asym_condn}, Theorem \ref{t:mainres1} also does not guarantee uniform stability in the sense of \cite[\S 2]{Liberzon}. 
	
	\begin{rem}
	\label{rem:compa_contra}
	\rm{
	The motivation behind the new result Theorem \ref{t:mainres1} is the purpose of identifying an umbrella framework for all classes of switching signals described in \S\ref{ss:the_res}. Although both Theorems \ref{t:asym_condn} and \ref{t:mainres1} deal solely with the asymptotic behaviour of the switching signals, the switching signals in Theorem \ref{t:asym_condn} afford a crisper characterization in terms of the properties of the switching signals in comparison to Theorem \ref{t:mainres1}. Indeed, in Theorem \ref{t:asym_condn} we have explicitly the asymptotic behaviour of various properties of the switching signals, viz., the switching frequency, frequency of admissible transitions, and the switching destinations. Clearly, by the properties of $\varliminf$ and $\varlimsup$ \cite[\S 0.2]{LojaReal}, the left-hand side of \eqref{e:maincondn2} is bounded above by
	\[
		\begin{aligned}
				& \varlimsup_{t\to+\infty}\nu(t)\sum_{(k,\ell)\in E(\P)}(\ln\mu_{k\ell})\varlimsup_{t\to+\infty}\rho_{k\ell}(t)\\
				& \qquad - \sum_{j\in\PS}\abs{\lambda_{j}}\varliminf_{t\to+\infty}\eta_{j}(t) + \sum_{k\in\PU}\abs{\lambda_{k}}\varlimsup_{t\to+\infty}\eta_{k}(t).
			\end{aligned}
	\]
	However, Theorem \ref{t:mainres1} is more general in the sense that it unifies all the existing characterizations of stabilizing switching signals that deal with both point-wise and asymptotic properties of the signals. 
This is the content of our next result. Interestingly enough, Theorem \ref{t:asym_condn} does not supply the unifying umbrella in this context precisely due to the ``crisper'' characterization described above; see our proof of Theorem \ref{t:mainres2} for a technical discussion. }
	\end{rem}

	\begin{theorem}
	\label{t:mainres2}
		Consider the family of systems \eqref{e:family}. Suppose that Assumptions \ref{assump:key} and \ref{assump:mu} hold. Then Theorem \ref{t:mainres1} unifies Theorems \ref{t:avg_dw_time}--{\ref{t:asym_condn}}.
	\end{theorem}

	{Recall that Theorems \ref{t:avg_dw_time}, \ref{t:mdavg_dw_time}, and \ref{t:avg_dw_time+u} provide point-wise characteristics of stabilizing switching signals, while Theorem \ref{t:asym_condn} characterizes stabilizing switching signals on the basis of their asymptotic properties.} In the light of Theorem \ref{t:mainres2}, it is clear that Theorem \ref{t:mainres1} unites all the above Theorems in terms of the asymptotic properties of the corresponding classes of switching signals.

	We provide detailed proofs of Theorems \ref{t:mainres1} and \ref{t:mainres2} in \S\ref{s:proofs}.
	
\section{Conclusion}
\label{s:concln}
	In this article we studied classes of stabilizing switching signals for continuous-time switched systems. Given a family of systems such that not all systems in the family are asymptotically stable, we proposed a new and general class of switching signals that recovers all existing results derived in the setting of multiple Lyapunov-like functions. Under standard assumptions, Theorem \ref{t:mainres1} extends to the discrete-time setting with minor modifications in the weights associated to the fraction of activation of subsystems $j\in\P$ until time $t > 0$ expressed by $\eta_{j}(t)$. Consequently, this extension recovers the discrete-time versions of the point-wise and asymptotic stability conditions presented in this article. We conjecture that the asymptotic stability condition for discrete-time switched systems presented in \cite{knc_hscc14} also follows from a discrete-time counterpart of Theorem \ref{t:mainres1}.

\section{Proofs}
\label{s:proofs}
	\begin{proof}[Proof of Theorem \ref{t:mainres1}]
		{Fix $t > 0$. Recall that $0=:\tau_{0}<\tau_{1}<\ldots<\tau_{\Ntsigma}$ are the switching instants before (and including) $t$. \ak{By} a straightforward iteration involving \eqref{e:Lyapprop} and \eqref{e:muijprop}, we obtain
		\begin{align}
		\label{e:pft6step1}
			V_{\sigma(t)}(x(t)) \leq \exp(\psi(t))V_{\sigma(0)}(x_{0})
		\end{align}
		with
		\begin{align}
		\label{e:psidefn}
			\psi(t) &:= \ln\Biggl(\prod_{i=0}^{\Ntsigma-1}\mu_{\sigma(\tau_{i})\sigma(\tau_{i+1})}\Biggr) - \sum_{\substack{i=0\\\tau_{\Ntsigma+1}:=t}}^{\Ntsigma}\lambda_{\sigma(\tau_{i})}S_{i}.
		\end{align}
		
		We have
		\begin{align}
			\ln\Biggl(\prod_{i=0}^{\Ntsigma-1}\mu_{\sigma(\tau_{i})\sigma(\tau_{i+1})}\Biggr) &= \sum_{i=0}^{\Ntsigma-1}\ln\mu_{\sigma(\tau_{i})\sigma(\tau_{i+1})} \nonumber\\
			&= \sum_{k\in\P}\sum_{i=0}^{\Ntsigma-1}\sum_{\substack{k\to\ell:\\\ell\in\P,\\k\neq\ell,\\\sigma(\tau_{i})=k,\\\sigma(\tau_{i+1})=\ell}}\ln\mu_{k\ell} \nonumber\\
			&= \sum_{(k,\ell)\in E(\P)}(\ln\mu_{k\ell})\mathrm{N}_{k\ell}(0,t) \nonumber\\
			\label{e:pft6step2} &\hspace*{-1cm}= \Ntsigma\sum_{(k,\ell)\in E(\P)}(\ln\mu_{k\ell})\rho_{k\ell}(t),
		\end{align}
		where $\rho_{k\ell}(t)$ is as defined in \eqref{e:tranfreq}.\\
		Also, $\displaystyle{-\sum_{\substack{i=0\\\tau_{\Ntsigma+1}:=t}}^{\Ntsigma}\lambda_{\sigma(\tau_{i})}S_{i} = -\sum_{\substack{i=0\\\tau_{\Ntsigma+1}:=t}}^{\Ntsigma}\sum_{j\in\P}{1}_{\{j\}}(\sigma(\tau_{i}))\lambda_{j}S_{i}}$.
		Separating out the asymptotically stable and unstable subsystems in the family \eqref{e:family}, we have that the right-hand side of the above equality is $\displaystyle{-\sum_{j\in\P_{S}}\lambda_{j}\Tj(0,t) - \sum_{k\in\P_{U}}\lambda_{k}\Tk(0,t)}$.
		By the properties of $\lambda_{j}$, the above expression can be rewritten as
		\begin{align}
		\label{e:pft6step3}
			-\sum_{j\in\P_{S}}\abs{\lambda_{j}}\Tj(0,t) + \sum_{k\in\P_{U}}\abs{\lambda_{k}}\Tk(0,t).
		\end{align}
		Replacing \eqref{e:pft6step2} and \eqref{e:pft6step3} in \eqref{e:psidefn}, we obtain
		\begin{align*}
			\psi(t) = \Ntsigma&\sum_{(k,\ell)\in E(\P)}(\ln\mu_{k\ell})\rho_{k\ell}(t) - \sum_{j\in\P_{S}}\abs{\lambda_{j}}\Tj(0,t)\\
			&+ \sum_{k\in\P_{U}}\abs{\lambda_{k}}\Tk(0,t). 
		\end{align*}
		For $t > 0$, the above expression can be written as
		\begin{align}
			\psi(t) &= t\Biggl(\frac{\Ntsigma}{t}\sum_{(k,\ell)\in E(\P)}(\ln\mu_{k\ell})\rho_{k\ell}(t) - \sum_{j\in\P_{S}}\abs{\lambda_{j}}\frac{\Tj(0,t)}{t}\nonumber\\
			&\hspace*{1cm}+ \sum_{k\in\P_{U}}\abs{\lambda_{k}}\frac{\Tk(0,t)}{t} 
			\Biggr) \nonumber\\
			\label{e:pft6step4}&= t\Biggl(\nu(t)\sum_{(k,\ell)\in E(\P)}(\ln\mu_{k\ell})\rho_{k\ell}(t) - \sum_{j\in\P_{S}}\abs{\lambda_{j}}\eta_{j}(t)\nonumber\\
			&\hspace*{1cm}+ \sum_{k\in\P_{U}}\abs{\lambda_{k}}\eta_{k}(t) 
			\Biggr),
		\end{align}
		where $\nu(t)$ and $\eta_{j}(t)$ are as defined in \eqref{e:swfreq} and \eqref{e:actifrac}, respectively.
		
		Now, by \eqref{e:keyprop1} and \eqref{e:pft6step1}, we obtain
		\begin{align}
		\label{e:pft6step5}
			\underline{\alpha}(\norm{x(t)}) \leq \exp(\psi(t))\overline{\alpha}(\norm{x_{0}}).
		\end{align}
		We verify GAS of the switched system \eqref{e:swsys} in two steps:
		\begin{enumerate}[leftmargin = *]
			\item i) we find conditions such that
			\begin{align}
			\label{e:toverify1}
				\lim_{t\to+\infty}\exp(\psi(t)) = 0,
			\end{align}
			ii) convergence is uniform for initial conditions $\tilde{x_{0}}$ satisfying $\norm{\tilde{x_{0}}}\leq\norm{x_{0}}$.
			\item we verify Lyapunov stability of \eqref{e:swsys} under any switching signal $\sigma$ that satisfies 1 i)--ii), i.e., it ensures uniform global asymptotic convergence of \eqref{e:swsys}.
		\end{enumerate}
		
		We begin with (1)i). Clearly, a sufficient condition for \eqref{e:toverify1} is that
		\begin{equation}
		\label{e:verd1}
		\begin{aligned}
			& \varlimsup_{t\to+\infty}\Biggl(\frac{\Ntsigma}{t}\sum_{(k,\ell)\in E(\P)}(\ln\mu_{k\ell})\rho_{k\ell}(t) - \sum_{j\in\P_{S}}\abs{\lambda_{j}}\eta_{j}(t)\\
			& \qquad + \sum_{k\in\P_{U}}\abs{\lambda_{k}}\eta_{k}(t) 
			\Biggr) < 0.
		\end{aligned}
		\end{equation}}
		
		We now move on to verify (1)ii). In view of \eqref{e:pft6step5}, we have
		\begin{align}
		\label{e:verd3}
			\norm{x(t)} \leq \underline{\alpha}^{-1}\biggl(\overline{\alpha}(\norm{x_{0}})\exp(\psi(t))\ak{\biggr)}\:\:\text{for all}\:\:t\geq 0.
		\end{align}
		Since the initial condition $x_{0}$ is decoupled from $\psi$ on the right-hand side of \eqref{e:verd3} and $\psi$ depends on $\sigma$, then for a fixed $\sigma$, if $\norm{x(t)}<\varepsilon$ for all $t > T(\norm{x_{0}},\varepsilon)$ for some pre-assigned $\varepsilon > 0$, then the solution $(\tilde{x}(t))_{t\geq 0}$ to \eqref{e:swsys} corresponding to an initial condition $\tilde{x_{0}}$ such that $\norm{\tilde{x_{0}}}\leq\norm{x_{0}}$ satisfies $\norm{\tilde{x}(t)} < \varepsilon$ for all $t > T(\norm{x_{0}},\varepsilon)$. Consequently, uniform global asymptotic convergence follows. Note that the uniformity here is over the initial condition \(x_0\), and not the set of switching signals \(\sigma\).
		
		It remains to verify (2). To this end, we need to show that for all $\varepsilon > 0$ there exists $\delta_{\varepsilon} > 0$ such that $\norm{x_{0}} < \delta_{\varepsilon}$ implies $\norm{x(t)} < \varepsilon$ for all $t\geq 0$. Fix $\varepsilon > 0$ and $\sigma\in\mathcal{S}$ such that $\sigma$ ensures uniform global asymptotic convergence of \eqref{e:swsys}. In other words, there exists $T(1,\varepsilon) > 0$ such that $\norm{x(t)} < \varepsilon$ for all $t > T(1,\varepsilon)$ whenever $\norm{x_{0}} < 1$.
		
		Let the family \eqref{e:family} be globally Lipschitz, and $L$ be the uniform Lipschitz constant over $\P$. It follows that with $\sigma\in\mathcal{S}$, $\norm{x(t)}\leq\exp(Lt)\norm{x_{0}}$ for all $t\geq 0$. Let $\delta' = \varepsilon\exp(-LT(1,\varepsilon))$. From the above inequality, it is evident that $\norm{x(t)}<\varepsilon$ for all $t\in[0,T(1,\varepsilon)]$ whenever $\norm{x_{0}} < \delta'$ with $\sigma\in\mathcal{S}$. To specialize to a $\sigma$ that ensures uniform global asymptotic convergence of \eqref{e:swsys}, we select $\delta = \min\{1,\delta'\}$, and Lyapunov stability of \eqref{e:swsys} follows at once. Now, if the family \eqref{e:family} is locally Lipschitz, we employ the following set of arguments to verify (2). Let $\varphi:[0,T(1,\varepsilon)]\:\lra\R$ be a function connecting $(\tau_{0},V_{\sigma(0)}(x_{0}))$, $(\tau_{i},\max\{y_{i},\tilde{y}_{i}\})$, $(T(1,\varepsilon),V_{\sigma(T(1,\varepsilon))}(x(T(1,\varepsilon))))$, $i = 1,2,\cdots,N$, where $N =$ number of switches before $T(1,\varepsilon), y_{i} = V_{\sigma(\tau_{i-1})}(x(\tau_{i}))$, and $\tilde{y}_{i} =$ $V_{\sigma(\tau_{i})}(x(\tau_{i}))$, with straight line segments. By construction, $\varphi$ is an upper envelope of $t\:\longmapsto V_{\sigma(t)}(x(t))$ on $[0,T(1,\varepsilon)]$, and is continuous. By continuity of $\varphi$ we have $\hat\varphi \Let \max_{t\in[0, T(1, \varepsilon)]}\varphi(t) < +\infty$. Also, due to \eqref{e:keyprop1}, $\hat\varphi\rightarrow 0$ as $\norm{x_{0}}\rightarrow 0$. It follows that there exists $\delta = \delta(\varepsilon) > 0$ such that whenever $\norm{x_{0}} < \delta(\varepsilon)$, we have $\hat\varphi < \varepsilon$.
		
		Our proof is now complete.
	\end{proof}	
	
	\ak{\begin{rem}
	\label{rem:addnl_condn}
		\ak{Observe that going one step beyond \eqref{e:verd1} and applying the properties \(\displaystyle{\varlimsup(\varphi_{1}+\varphi_{2})\leq\varlimsup\varphi_{1}+\varlimsup\varphi_{2}}\), \(\displaystyle{\varliminf(\varphi_{1}+\varphi_{2})\geq\varliminf\varphi_{1}+\varliminf\varphi_{2}}\), one obtains \eqref{e:asympcondn2}.} 
	\end{rem}}

	\begin{proof}[Proof of Theorem \ref{t:mainres2}]
		{Theorem \ref{t:avg_dw_time} follows as a special case of Theorem \ref{t:mdavg_dw_time} with $\lambda_{s} = \displaystyle{\min_{j\in\P_{S}}\abs{\lambda_{j}}}$, $\displaystyle{\mu = \max_{j\in\P}\mu_{j}}$, $\displaystyle{\mathrm{N}_{0} = \sum_{j\in\P}\mathrm{N}_{0_{j}}}$, $\displaystyle{\frac{1}{\tau_{a}} = \sum_{j\in\P}\frac{1}{\tau_{a}^{j}}}$. Moreover, Theorem \ref{t:avg_dw_time} follows as a special case of Theorem \ref{t:avg_dw_time+u} when $\P_{U}= \emptyset$.}

		Therefore, in order to show that Theorem \ref{t:mainres1} unifies Theorems \ref{t:avg_dw_time}, \ref{t:mdavg_dw_time}, \ref{t:avg_dw_time+u} and \ref{t:asym_condn} under an umbrella framework, it suffices to show the following: if a switching signal $\sigma$ satisfies the conditions in Theorem \ref{t:mdavg_dw_time} (resp. Theorems \ref{t:avg_dw_time+u}, and \ref{t:asym_condn}), then the conditions in Theorem \ref{t:mainres1} follow, and by the assertion of Theorem \ref{t:mainres1}, the switched system \eqref{e:swsys} is GAS.
		
		(I) {We first show that if a \(\sigma\) satisfies the conditions in Theorem \ref{t:mdavg_dw_time}, then the conditions in Theorem \ref{t:mainres1} follow.}
		
		Assume that a switching signal $\sigma\in\mathcal{S}$ satisfies mode-dependent average dwell time $\tau_{a}^{j}$ such that \eqref{e:mdadtcondn1} holds. 
		{It suffices to show that the above \(\sigma\) satisfies \eqref{e:maincondn2}.}
		

		
		We have for any $t > 0$,
		\begin{align*}
			&\nu(t)\sum_{(i,j)\in E(\P)}(\ln\mu_{ij})\rho_{ij}(t) - \sum_{j\in\P_{S}}\abs{\lambda_{j}}\eta_{j}(t)\nonumber\\ &= \sum_{j\in\P}(\ln\mu_{j})\frac{\mathrm{N}_{{j}}(0,t)}{t} - \sum_{j\in\P}\lambda_{j}\eta_{j}(t)
		\end{align*}
		By definition of mode-dependent average dwell time, the right-hand side of the above quantity is bounded above by
		\[
			\sum_{j\in\P}(\ln\mu_{j})\frac{\mathrm{N}_{0}^{j}}{t} + \sum_{j\in\P}\frac{(\ln\mu_{j})}{\tau_{a}^{j}}\frac{\Tj(0,t)}{t} - \sum_{j\in\P}\lambda_{j}\frac{\Tj(0,t)}{t}.
		\]
		In view of \eqref{e:mdadtcondn1}, the above expression is at most equal to
		\begin{align*}
			&\sum_{j\in\P}(\ln\mu_{j})\frac{\mathrm{N}_{0}^{j}}{t} + \sum_{j\in\P}(\ln\mu_{j})\frac{\Tj(0,t)}{t}\Biggl(\frac{\lambda_{j}}{\ln\mu_{j}}-\varepsilon_{j}\Biggr) \\&- \sum_{j\in\P}\lambda_{j}\frac{\Tj(0,t)}{t}\:\text{with}\:\varepsilon_{j} > 0\:\text{for all}\:j\in\P\\
			&= \sum_{j\in\P}(\ln\mu_{j})\frac{\mathrm{N}_{0}^{j}}{t} - \sum_{j\in\P}\varepsilon_{j}(\ln\mu_{j})\frac{\Tj(0,t)}{t}.
		\end{align*}
		Therefore,
		\begin{align*}
			&\varlimsup_{t\to+\infty}\Biggl(\nu(t)\sum_{(i,j)\in E(\P)}(\ln\mu_{ij})\rho_{ij}(t) - \sum_{j\in\P_{S}}\abs{\lambda_{j}}\eta_{j}(t)\Biggr)\\ &\leq \varlimsup_{t\to+\infty}\Biggl(\sum_{j\in\P}(\ln\mu_{j})\frac{\mathrm{N}_{0}^{j}}{t} - \sum_{j\in\P}\varepsilon_{j}(\ln\mu_{j})\frac{\Tj(0,t)}{t}\Biggr)\\
			&\leq\varlimsup_{t\to+\infty}\Biggl(\sum_{j\in\P}(\ln\mu_{j})\frac{\mathrm{N}_{0}^{j}}{t}\Biggr) - \varliminf_{t\to+\infty}\Biggl(\varepsilon_{j}(\ln\mu_{j})\frac{\Tj(0,t)}{t}\Biggr)\\
			&\leq - \sum_{j\in\P}\varepsilon_{j}(\ln\mu_{j})\varliminf_{t\to+\infty}\frac{\Tj(0,t)}{t} < 0\:\text{since}\:\mu_{j} > 1\:\text{for all}\:j\in\P.
		\end{align*}
		Consequently, \eqref{e:maincondn2} holds, and by the assertion of Theorem \ref{t:mainres1}, the switched system \eqref{e:swsys} is GAS. 

		(II) {We now show that if a \(\sigma\) satisfies the conditions in Theorem \ref{t:avg_dw_time+u}, then the conditions in Theorem \ref{t:mainres1} follow.}
		
		Assume that a switching signal $\sigma\in\mathcal{S}$ satisfies average dwell time $\tau_{a}$ such that \eqref{e:un_durn} and \eqref{e:ADT+u} hold.
		
		
		It suffices to show that the \(\sigma\) under consideration satisfies \eqref{e:maincondn2}. We have
		\begin{align}
		\label{e:pft2step2}
			&\varlimsup_{t\to+\infty}\Biggl(\nu(t)\sum_{(k,\ell)\in E(\P)}(\ln\mu_{k\ell})\rho_{k\ell}(t) - \sum_{j\in\PS}\abs{\lambda_{j}}\eta_{j}(t)\nonumber\\ &\hspace*{1cm}+ \sum_{k\in\PU}\abs{\lambda_{k}}\eta_{k}(t)\Biggr)\nonumber\\
			&\leq \varlimsup_{t\to+\infty}\nu(t)\varlimsup_{t\to+\infty} \Biggl(\sum_{(k,\ell)\in E(\P)}(\ln\mu_{k\ell})\rho_{k\ell}(t)\Biggr)\nonumber\\&\hspace*{1cm}+ \varlimsup_{t\to+\infty}\Biggl(- \sum_{j\in\PS}\abs{\lambda_{j}}\eta_{j}(t) + \sum_{k\in\PU}\abs{\lambda_{k}}\eta_{k}(t)\Biggr).
		\end{align}
		Firstly, since $\sigma$ satisfies average dwell time $\tau_{a}$, we have that
		\begin{align}
		\label{e:pft2step3}
			\varlimsup_{t\to+\infty}\nu(t) \leq \varlimsup_{t\to+\infty}\Biggl(\frac{\mathrm{N}_{0}}{t}+\frac{1}{\tau_{a}}\Biggr) \leq \frac{1}{\tau_{a}},
		\end{align}
		and in view of $\displaystyle{\Ntsigma = \sum_{(k,\ell)\in E(\P)}\mathrm{N}_{k\ell}(0,t)}$, we have
		\begin{align}
		\label{e:pft2step4}
			\varlimsup_{t\to+\infty}\Biggl(\sum_{(k,\ell)\in E(\P)}(\ln\mu_{k\ell})\rho_{k\ell}(t)\Biggr) = \ln\mu.
		\end{align}
		Secondly, applying $\displaystyle{\TS(0,t) = \sum_{j\in\PS}\Tj(0,t)}$, \\$\displaystyle{\TU(0,t) = \sum_{k\in\PU}\Tk(0,t)}$, and $t = \TS(0,t) + \TU(0,t)$, we get
		\begin{align*}
			&\varlimsup_{t\to+\infty}\Biggl(- \sum_{j\in\PS}\abs{\lambda_{j}}\eta_{j}(t) + \sum_{k\in\PU}\abs{\lambda_{k}}\eta_{k}(t)\Biggr) \\&= \varlimsup_{t\to+\infty}\Biggl(-\lambda_{s}\frac{(t-\TU(0,t))}{t} + \lambda_{u}\frac{\TU(0,t)}{t}\Biggr)\\
			&\leq -\lambda_{s} + (\lambda_{s}+\lambda_{u})\varlimsup_{t\to+\infty}\frac{\TU(0,t)}{t}.
		\end{align*}
		In view of \eqref{e:un_durn}, the above expression is at most equal to
		\begin{align}
		\label{e:pft2step5}
			-\lambda_{s} + (\lambda_{s}+\lambda_{u})\varlimsup_{t\to+\infty}\Biggl(\frac{\mathrm{T}_{0}}{t} + \rho\Biggr) \leq -\lambda_{s} + (\lambda_{s}+\lambda_{u})\rho.
		\end{align}
		Replacing \eqref{e:pft2step3}--\eqref{e:pft2step5} in \eqref{e:pft2step2}, we obtain
		\begin{align*}
			&\varlimsup_{t\to+\infty}\Biggl(\nu(t)\sum_{(k,\ell)\in E(\P)}(\ln\mu_{k\ell})\rho_{k\ell}(t) - \sum_{j\in\PS}\abs{\lambda_{j}}\eta_{j}(t)\nonumber\\&\hspace*{1cm}+ \sum_{k\in\PU}\abs{\lambda_{k}}\eta_{k}(t)\Biggr)\\
			&\leq\frac{1}{\tau_{a}}(\ln\mu)-\lambda_{s}(1-\rho)+\lambda_{u}\rho.
		\end{align*}
		In view of \eqref{e:ADT+u}, the above quantity is bounded above by $-\varepsilon(\ln\mu)$ for some $\varepsilon > 0$. Since $\mu > 1$, \(-\varepsilon (\ln\mu)\) is strictly smaller than \(0\). Consequently, \eqref{e:maincondn2} holds, and by the assertion of Theorem \ref{t:mainres1}, we conclude that the switched system \eqref{e:swsys} is GAS. 
		
			Observe that the set of arguments in (I) and (II) 
			do not follow from \eqref{e:asympcondn2} because of the following properties of $\varliminf$ and $\varlimsup$ \cite[\S0.2]{LojaReal}:
			\begin{align*}
				\varlimsup(\varphi_{1}+\varphi_{2})\leq\varlimsup\varphi_{1}+\varlimsup\varphi_{2}\\
				\varliminf(\varphi_{1}+\varphi_{2})\geq\varliminf\varphi_{1}+\varliminf\varphi_{2}
			\end{align*}
			that hold whenever the right-hand sides are not of the form $\mp\infty\pm\infty$. Consequently, Theorem \ref{t:asym_condn} does not offer an umbrella framework for Theorems \ref{t:avg_dw_time}-\ref{t:avg_dw_time+u}; for that we need Theorem \ref{t:mainres1}.
		
		(III) {We finally show that if a \(\sigma\) satisfies the conditions in Theorem \ref{t:asym_condn}, then the conditions in Theorem \ref{t:mainres1} follow.}
		
		{Assume that a switching signal $\sigma$ satisfies \eqref{e:asympcondn2}. 
		We demonstrate that \eqref{e:asympcondn2} implies \eqref{e:maincondn2}.}
		\begin{align}
		\label{e:pft2step6}
			&\varlimsup_{t\to+\infty}\Biggl(\nu(t)\sum_{(k,\ell)\in E(\P)}(\ln\mu_{k\ell})\rho_{k\ell}(t) - \sum_{j\in\PS}\abs{\lambda_{j}}\eta_{j}(t) \nonumber\\&\hspace*{1cm}+ \sum_{k\in\PU}\abs{\lambda_{k}}\eta_{k}(t)\Biggr)\nonumber\\
			&\leq \varlimsup_{t\to+\infty}\nu(t)\varlimsup_{t\to+\infty} \Biggl(\sum_{(k,\ell)\in E(\P)}(\ln\mu_{k\ell})\rho_{k\ell}(t)\Biggr)\nonumber\\&\hspace*{1cm} - \varliminf_{t\to+\infty}\Biggl(\sum_{j\in\PS}\abs{\lambda_{j}}\eta_{j}(t)\Biggr) + \varlimsup_{t\to+\infty}\Biggl(\sum_{k\in\PU}\abs{\lambda_{k}}\eta_{k}(t)\Biggr)\nonumber\\
			&\leq \varlimsup_{t\to+\infty}\nu(t)\sum_{(k,\ell)\in E(\P)}(\ln\mu_{k\ell})\varlimsup_{t\to+\infty}\rho_{k\ell}(t)\nonumber\\&\hspace*{1cm} - \sum_{j\in\PS}\abs{\lambda_{j}}\varliminf_{t\to+\infty}\eta_{j}(t) + \sum_{k\in\PU}\abs{\lambda_{k}}\varlimsup_{t\to+\infty}\eta_{k}(t).
		\end{align}
		In view of \eqref{e:asympcondn2}, the right-hand side of \label{e:pft2step6} is strictly less than 0. Hence, \eqref{e:maincondn2} follows and by the assertion of Theorem \ref{t:mainres1}, the switched system \eqref{e:swsys} is GAS under the switching signal $\sigma$ in discussion. To wit, Theorem \ref{t:avg_dw_time+u} follows from Theorem \ref{t:mainres1}.

		This completes our proof of Theorem \ref{t:mainres2}.
	\end{proof}

	\begin{rem}
	\label{r:expo}
		Notice that all the prior results that are unified in the framework of Theorem \ref{t:mainres1} relate to exponential convergence (a quantitative property) of the function \(t\mapsto V_{\sigma(t)}(x(t))\) in terms of the notation established above. Indeed, the left-hand side of {\eqref{e:maincondn2}} is at most equal to $-c$ for some scalar $c>0$, which in conjunction with \eqref{e:keyprop1} and \eqref{e:pft6step5} ensures asymptotically exponential convergence rate of the Lyapunov-like functions along system trajectories. In contrast, the recent work \cite{ghi}, geared towards input-to-state stability (ISS) of switched systems,\footnote{Recall that if the input is set to zero for all time, then the ISS property reduces to the GAS property.} characterizes stability in terms of certain class $\mathcal{FK}_{\infty}$ functions, and the only essential property imposed there is monotonicity (a purely qualitative property).\footnote{$\mathcal{FK}_{\infty} := \{\varrho:[0,+\infty[^{2}\to[0,+\infty[^{2}\:|\:\varrho\:\text{is continuous, and for every fixed first argument},\:\varrho\in\mathcal{K}_{\infty}\:\text{in the second argument}\}.$} Consequently, the conditions for stabilizing switching signals that can be obtained via \cite{ghi} will not fall out as a special case of Theorem \ref{t:mainres1} that guarantees exponential convergence. However, qualitative results can be specialized to give quantitative ones, e.g., in \cite{ghi} the authors showed that prior results based on ISS under average dwell time switching follow as non-trivial special cases of the results in \cite{ghi}.
	\end{rem}

%
%
\section*{Appendix}
	Fix \(t>0\). In \cite{abc} the quantity \(\eta_{j}(t)\) was defined on the interval \(]0,\tau_{\Ntsigma}]\) and \eqref{e:asympcondn2} was derived under the condition that 
	\begin{align}
	\label{e:app_condn1}
		\lim_{t\to+\infty}\frac{t-\tau_{\Ntsigma}}{t} = 0.
	\end{align}
	In other words, as \(t\to+\infty\) the duration of activation of any system in \eqref{e:family} is not comparable to \(t\). To arrive at \eqref{e:app_condn1}, the hypothesis
	\begin{align}
	\label{e:app_condn2}
		\varliminf_{t\to+\infty}\nu(t) > 0
	\end{align}
	was employed. However, the claim that \eqref{e:app_condn2} implies \eqref{e:app_condn1} was incorrect. Indeed, consider a switching signal \(\sigma\) of the following nature (by the phrase ``immediately after'' that appears below, we refer to a small interval of length \(\varepsilon > 0\)):
	\begin{itemize}[label = \(\circ\), leftmargin = *]
		\item immediately after \(t=0\), there is \(1\) switch and no further switches till \(t=1\)
		\item immediately after every \(t=2^n\), there is/are \(2^n\) switch(es) and no further switches till \(t=2^{n+1}\), \(n\in\N_{0}\).
	\end{itemize}
	Now, fix any \(n>0\) and \(2^n+\varepsilon<t\leq 2^{n+1}\), \(t\) large enough. We have that \(\displaystyle{1+\sum_{k=0}^{n}2^k = 2^{n+1}}\). Consequently, \(\ak{\displaystyle{\displaystyle{\varliminf_{t\to+\infty}}\nu(t)\geq 1}}\) but as \(t\to+\infty\), \(t-\tau_{\Ntsigma}\) is not negligible compared to \(t\). 
	
	\ak{A careful observation, however, reveals that the condition \eqref{e:app_condn1} may be weakened to}
	\begin{align}
	\label{e:app_condn3}
		\varlimsup_{t\to+\infty}\max_{\sigma(\tau_{\Ntsigma})\in\P_{U}}\frac{t-\tau_{\Ntsigma}}{t} = 0;
	\end{align}
	\eqref{e:app_condn3} requires that as \(t\to+\infty\), a switching signal does not \emph{dwell} on the unstable subsystems for time durations comparable to \(t\). In fact, it also provides a quantitative estimate that \(\displaystyle{\frac{t-\tau_{\Ntsigma}}{t}} = o(t)\) as \(t\to+\infty\) and \(\sigma(\tau_{\Ntsigma})\in\P_{U}\). Indeed, \emph{dwell}ing on an asymptotically stable subsystem for longer time durations (comparable to \(t\) as \(t\to+\infty\)) does not act \emph{against} our objective to stabilize \eqref{e:swsys}.  
	
	A further careful observation leads us to the fact that with the quantity \(\eta_{j}(t)\) defined on \(]0,t]\) in place of \(]0,\tau_{\Ntsigma}]\), only condition \eqref{e:asympcondn2} is sufficient to guarantee GAS of \eqref{e:swsys}, see Remark \ref{rem:addnl_condn}. Indeed, as \(t\to+\infty\) long durations of activation for asymptotically stable subsystems and short durations of activation for unstable subsystems are perfectly admissible so long as \eqref{e:asympcondn2} is satisfied. The condition \eqref{e:asympcondn2} includes \(\displaystyle{\varlimsup_{t\to+\infty}\nu(t)}\) and \(\displaystyle{\varliminf_{t\to+\infty}\nu(t)\leq\varlimsup_{t\to+\infty}\nu(t)}\). However, no lower bound on \(\displaystyle{\varliminf_{t\to+\infty}\nu(t)}\) is required for condition \eqref{e:asympcondn2} to hold.
\section*{Acknowledgements}
	We are grateful to the anonymous reviewers whose comments helped us to significantly improve the main results and the overall quality of the manuscript.



	\end{multicols}

\bigskip

\end{document}